\documentclass[reqno,11pt]{amsart}
\usepackage{amsmath,amssymb,graphics,epsfig,color,psfrag}
\usepackage{verbatim}
\newtheorem{theorem}{Theorem}[section]

\newtheorem{definition}[theorem]{Definition}
\newtheorem{lemma}[theorem]{Lemma}

\newtheorem{fact}[theorem]{Fact}

\numberwithin{equation}{section} \numberwithin{theorem}{section}

\renewcommand{\epsilon}{\varepsilon}

\definecolor{light}{gray}{.9}
\newcommand{\cA}{\ensuremath{\mathcal A}}

\newcommand{\cC}{\ensuremath{\mathcal C}}

\newcommand{\cE}{\ensuremath{\mathcal E}}
\newcommand{\cF}{\ensuremath{\mathcal F}}
\newcommand{\cG}{\ensuremath{\mathcal G}}

\newcommand{\cL}{\ensuremath{\mathcal L}}

\newcommand{\cS}{\ensuremath{\mathcal S}}
\newcommand{\cT}{\ensuremath{\mathcal T}}

\newcommand{\bbE}{{\ensuremath{\mathbb E}} }

\newcommand{\bbI}{{\ensuremath{\mathbb I}} }

\newcommand{\bbN}{{\ensuremath{\mathbb N}} }

\newcommand{\bbP}{{\ensuremath{\mathbb P}} }

\newcommand{\bbR}{{\ensuremath{\mathbb R}} }

\newcommand{\bbZ}{{\ensuremath{\mathbb Z}} }

\let\a=\alpha \let\b=\beta   \let\d=\delta  \let\e=\varepsilon
 \let\g=\gamma       \let\l=\lambda
      \let\o=\omega      
  \let\s=\sigma    
  
\let\D=\Delta   \let\G=\Gamma   
\let\O=\Omega      

\author [A.\ Faggionato]{Alessandra Faggionato}
\address{Alessandra Faggionato,
  Dipartimento di Matematica, Universit\`a di Roma ``La Sapienza",
  P.le Aldo Moro 2, 00185 Roma, Italy}
\email{faggiona@mat.uniroma1.it}

\author [D.\ Di Pietro]{Daniele Di Pietro}
\address{Daniele Di Pietro,
  Dipartimento di Fisica, Universit\`a di Roma ``La Sapienza",
  P.le Aldo Moro 2, 00185 Roma, Italy}

\title[Gallavotti--Cohen--Type symmetries for Markov chains]{Gallavotti--Cohen--Type symmetry  related to cycle decompositions for
Markov chains and biochemical applications}

\thanks{Work partially supported by the European Research Council through the ``Advanced
Grant'' PTRELSS 228032}

\begin{document}

\maketitle

\begin{abstract}
We  slightly extend  the fluctuation theorem   obtained in \cite{LS}
for sums of generators,  considering continuous--time Markov chains
on a finite state space whose underlying graph has multiple edges
and no loop. This extended frame is suited when analyzing chemical
 systems. As  simple corollary we derive  by  a different method the fluctuation
 theorem of  D.\ Andrieux and P.\ Gaspard
 for the fluxes along the chords  associated to a fundamental set of
 oriented cycles
\cite{AG2}.

We associate to each random trajectory an oriented cycle on the
graph and we decompose it in terms of a basis of oriented cycles. We
prove  a  fluctuation theorem for the coefficients in this
decomposition.
 The resulting fluctuation theorem involves the cycle
affinities, which in many real systems correspond to the macroscopic
forces. In addition, the above decomposition is useful when
analyzing the large deviations of additive functionals of the Markov
chain.
  As example of application, in a very
general context we derive a fluctuation relation for  the mechanical
and chemical currents of a molecular motor moving along a periodic
filament.

\smallskip

\smallskip


\end{abstract}

\section{Introduction}
 Out--of--equilibrium systems are common in daily life. Examples are mechanical systems in contact with thermal
 reservoirs as well reacting systems in contact with particle reservoirs
 generating particle fluxes through differences of chemical
 potentials. Considering chaotic dynamical systems of statistical mechanics, Gallavotti and Cohen \cite{GC} have discovered a symmetry relation (now
 taking their names) for the large deviation functional of the average entropy creation rate. This relation is
 also called {\sl fluctuation theorem}.  The same result has then been derived for Markov stochastic processes by
 Kurchan \cite{K}, Maes \cite{M} and Lebowitz and Spohn \cite{LS}. An extensive rigorous  treatment including  further developments is given
 in the book of Jiang et al. \cite{JQQ} (see also references
 therein together with \cite{BG}).
  Near to equilibrium, the   fluctuation theorem generalized to systems with several currents
  yields the Onsager's  symmetry and the usual Green--Kubo's  formulas for
transport coefficients \cite{G,LS}. In this sense, it can be thought
as their generalization far from equilibrium.

  The initial investigation of the fluctuation theorem
  has referred to models coming from statistical mechanics, and in the last years a proper analysis
 has been developed for chemical and biochemical systems by  Andrieux and Gaspard
 (see for example \cite{Ga,AG0,AG1,AG2,AG3,AG4} and references therein). The analysis for a model of molecular motor
  along a periodic filament with two chemical states has also been done  by  Lacoste, Lau, Mallick \cite{LM,LLM0,LLM1,LM1}.

The approach of Andrieux and Gaspard is much inspired by  the
network theory of out--of--equilibrium systems. This theory has been
 developed  by Hill \cite{H} and Schnakenberg \cite{SCH} with a special attention to biochemical systems. Mathematically, one is interested to  the stochastic
 dynamics of a continuous time random walk on a finite  connected graph $\cG$ with multiple edges and no loop, where the jumps
 along an edge
  have positive probability rate  in both directions. For such a random walk the stationary distribution $\mu$ exists and is unique. To each
   oriented bond $\ell$ in $\cG$ from
  the state $\s$ to the state $\s'$, one associates a microscopic affinity defined as
  \begin{equation}
 \cA(\ell)= \ln \frac{ \mu(\s) k(\s, \s')}{\mu(\s') k(\s',
\s)}\,,
  \end{equation} where $k(\s_1,\s_2)$ denotes the probability rate for a jump from
$\s_1$ to $\s_2$. Note that  detailed balance in reversible systems
simply corresponds to the fact that all microscopic affinities are
zero. The affinity associated to a given oriented cycle $\cC$  in
$\cG$ is then defined as the sum  of the microscopic affinities of
the oriented bonds forming the cycle: $\cA(\cC)=\sum _{\ell \in
\cC}\cA(\ell)$. It is simple to check that the affinity of a cycle
remains the same if one replaces the above microscopic affinity
$\cA(\ell)$ by $\ln \bigl(k(\s, \s')/ k(\s', \s)\bigr)$.

The network theory presented in \cite{SCH} is based on the fact that
to each unoriented spanning (maximal) tree on the graph $\cG$ one
can associate in a canonical way a family of oriented cycles, called
{\sl fundamental set of oriented cycles}, being a sort of basis of
the space of all oriented cycles. Each oriented cycle in a
fundamental set contains only one oriented edge (called {\sl chord})
that does not belong to the spanning tree when disregarding
orientation. In \cite{SCH} Schnakenberg has observed that, for
several models of  real systems, the macroscopic (mechanical or
thermodynamical) forces keeping the system out--of--equilibrium
 are the affinities  of the
oriented cycles of some fundamental set  in $\cG$ (note that more
oriented cycles can have the same affinity). In \cite{AG2} this
situation is called {\sl Schnakenberg condition}. We point out (see
\cite{SCH,AG2}) that usually thermodynamic forces are not encoded in
a  single jump rate.

 In \cite{AG2}, starting from a spanning tree,  Andrieux and Gaspard have derived a fluctuation theorem for the currents
 along the chords, whose conjugate
 variables are the affinities  associated to the oriented cycles in the
 fundamental set. Under the Schnakenberg condition, this fluctuation
 theorem can be restated in terms of the macroscopic forces keeping
 the system out--of--equilibrium instead of cycle affinities.

Let us explain our theoretical contribution in this direction.
First, we recover the above results of Andrieux and Gaspard by a
different approach. Indeed, the fluctuation theorem presented for
sums of several generators in Section 2.3 of \cite{LS} is not suited
for random walks on graphs with multiple edges, but as we show its
simple proof can be easily adapted to the present case. This
slightly extended result immediately leads to the fluctuation
theorem for the family of currents along all oriented edges, and
therefore to the fluctuation theorem for the  currents referred to a
fundamental set of oriented cycles.
 We point out that the
fundamental sets of oriented cycles, canonically associated to the
spanning trees of the graph $\cG$, do not cover all possible bases
of the cycle space.  We  then  extend the network theory presented
in \cite{SCH} and  prove a fluctuation theorem for the ``currents"
referred to a generic basis of oriented cycles. When dealing with a
generic basis, chords disappear and  the definition of the variables
conjugated to the affinities of the oriented cycles in the basis is
more algebraic. We call these conjugate variables {\sl generalized
currents}. Their definition needs to associate to each random
trajectory an oriented cycle, which must then be decomposed in the
basis. The generalized currents are the coefficients in this
decomposition. Their definition is essentially  algebraic and
differs from the {\sl derived chain}  discussed in
\cite{JQQ}[Chapters 1,2]. In the derived chain method, for each
oriented cycle $\cC$  one counts the number of times $\o_t(\cC)$
that the cycle has been realized by the trajectory up to a given
time $t$ (there is nothing algebraic): in \cite{JQQ} the authors
show that this quantity, rescaled by $t$, converges a.s. to some
circulation number $\o(\cC)$ and show that the entropy production
rate in the steady state can be expressed in terms of
 the circulation numbers. For a discussion on fluctuation theorems
satisfied by $\o_t (\cC)$ we refer to \cite{AG4}.


The above cycle decomposition of the trajectory is particularly
useful when looking for a fluctuation theorem of additive
functionals of the Markov chain. As example of application in this
direction, we consider a very general discrete model for a molecular
motor moving along a periodic filament under the action of an
external force $f$, transforming the chemical energy from ATP
hydrolysis to mechanical work. The forces keeping the system
out--of--equilibrium are the external force $f$ and the chemical
potential difference $\D \mu$ associated to ATP hydrolysis.
 In the last years, much attention have been devoted to
 the thermodynamics of small systems supported by very fast  technological
 improvements  (see for example \cite{BLR,Ri} and references therein).  Molecular motors
 are  proteins  working as motors at the nanoscale,  with very interesting thermodynamic aspects, which have  been
 much investigated both  from a theoretical and an experimental viewpoint \cite{Ho,JAP,KF,Re}.
 The fluctuation theorem for molecular motors with two or three chemical states have been proved in
 \cite{AG0} and \cite{LM,LLM0,LLM1,LM1}. The small number
 of chemical states allows a detailed analysis based on matrix computations, leading to more information than the
 Gallavotti--Cohen--type symmetry for large deviations. In the general case,  matrix computations
 become not reasonable. On the other hand, the developed theory for
 fluctuation theorems associated to cycle decompositions allows to
 easily prove in full generality the fluctuation relation obtained in \cite{LLM0,LLM1}
  for only two chemical
 states. More precisely, calling $x_t$ the position of the molecular motor along the filament and $z_t$
  the ATP consumption, it holds
\begin{equation}\label{fatto}
 \vartheta ( \l, \g) = \vartheta (\overline{f\b} - \l, \b\D\mu -\g)\,, \qquad \l , \g \in \bbR,
 \end{equation}
where, roughly speaking, the function  $\vartheta (\l, \g)$ is
characterized  by the identity
   \begin{equation}\text{$e^{- \vartheta (\l, \g) t }\sim \bbE( e^{- \l \bar x_t - \g z_t})$}\,, \qquad  t  \gg
   1\,,
   \end{equation}
    and where
 $\overline{f\b}$  and $\bar x_t$ denote the dimensionless quantities $f \b
\times [1m]$ and $x_t /[1m]$ ($[1m]$ being the length unit). Of
course, $\b= 1/ kT$ where  $k$ is the Boltzmann's constant and $T$
the absolute temperature.
\subsection{Outline of the paper} In Section \ref{pierpi} we extend
the fluctuation theorem for sums of generators (see Theorem
\ref{fluttuante}). In Section \ref{corrente} we discuss some
fluctuation theorems  (Fact \ref{robots} and Theorem \ref{basis})
related to cycle decompositions (fundamental set of oriented cycles
and basis, respectively). All proofs are postponed to Sections
\ref{integ_pierpi}, \ref{chiudo}, \ref{integ_FT_basis}. Finally, in
Section \ref{motorino} we discuss the biochemical application of
Theorem \ref{basis} leading to relation \eqref{fatto}.

\section{Fluctuation Theorem for sums of Markov generators}\label{pierpi}
As starting point, we extend the fluctuation theorem obtained in \cite{LS}[Section 2.3].
  As in \cite{LS} we consider a continuous--time Markov chain
$(X_t)_{t\geq 0}$ on a finite state space $\cS$ having Markov
generator
\begin{equation}\label{generatore} \cL f (\s) = \sum _{\s': \, \s' \not = \s} k(\s,
\s') \bigl( f(\s')-f(\s) \bigr)\,, \qquad \s \in
\cS\,.\end{equation} We assume that  the Markov chain is irreducible
and that  the transition rates satisfy the positivity relation
\begin{equation}\label{pr}
k(\s, \s')>0 \; \Longleftrightarrow\; k(\s', \s) >0\,.
\end{equation}
The last  assumption is very reasonable in physical systems,  since
for a real transition between physical states one expects that, at
least with very low probability, the opposite transition is
possible. We remark that, due to irreducibility and since $\cS$ is
finite, there exists a unique stationary distribution  of the
Markov chain.

\medskip

We decompose the generator $\cL$ as a sum of $m$ generators
$\cL^{(i)}$. More precisely, we fix  a family of numbers
$k^{(i)}(\s, \s')$, parameterized by the index  $i=1, \dots, m $ and
by the   pairs $(\s, \s')$ of distinct elements in $\cS$. We assume
that: \begin{itemize}

\item[(A1)]  $ k^{(i)}(\s, \s') \geq 0$ for all $i:1\leq i \leq m $ and
$\s\not  =\s'$ in $\cS$;

\item[(A2)]  $\sum_{i=1}^m k ^{(i)} (\s, \s')= k (\s,\s')$ for all
$\s\not  =\s'$ in $\cS$;

\item[(A3)]  $k^{(i)}(\s,
\s')>0 \; \Longleftrightarrow\; k^{(i)}(\s', \s) >0$, for all
$i:1\leq i \leq m$ and $\s\not =\s'$ in $\cS$.
\end{itemize}
Trivially, the above family $\bigl\{ k^{(i)}(\s,\s')\bigr\}
_{i,\s,\s'}$ corresponds to a representation of  $\cL$ as sum of
Markov generators, since we can write $\cL= \sum _{i=1}^m \cL^{(i)}$
where $\cL^{(i)}$ denotes the Markov generator on $\cS$ defined as
$\cL^{(i)} f(\s)= \sum _{\s': \s' \not =\s } k^{(i)}(\s, \s') \bigl(
f(\s')- f(\s) \bigr)$.

\smallskip

We point out that  our assumptions (A1), (A2), (A3) coincide with
the ones stated   in \cite{LS}[Section 2.3], with the only exception
that here we have dropped the additional assumption of \cite{LS}
that, given $\s \not = \s'$, there exists at most one index $i$ such
that $k^{(i)}(\s,\s')>0$.

\smallskip

The above construction suggests to think of  the Markov chain as a
random walk on a connected graph $\cG$ with multiple edges and no
loop: the vertexes of $\cG$ are given by the states of $\cS$, while
two distinct states $\s, \s'$ are linked by as many edges as the
indexes $i$ for which $k^{(i)}(\s, \s'), k^{(i)}(\s',\s)>0$, each
edge is labeled by the corresponding index $i$.  Then, $k^{(i)} (\s,
\s')$ represents the probability rate for a jump along the
$i$--labeled edge from $\s$ to $\s'$. Such an  interpretation of the
Markov chain as random walk on a graph with multiple edges   becomes
particularly relevant in biochemical applications and is essential
when considering thermodynamical forces and affinities
\cite{H,SCH,Ga,AG2}.

\smallskip

 We formulate the fluctuation theorem in terms of  the random walk on the
graph $\cG$ presented above. By a trajectory up to time $t$, we mean the path  $(X_s: s \in [0,t])$ together with the
knowledge of the edges along which   the walker has moved.
Given $i:1
\leq i \leq m$ and states $\s\not = \s'$, we define the {\sl weight}
\begin{equation}
 w^{(i)}(\s, \s'):= \begin{cases}
 \ln \frac{k^{(i)}(\s, \s') }{ k^{(i)}(\s', \s)} & \text{ if }
 k^{(i)}(\s,\s'), k^{(i)}(\s', \s)>0\,;\\
 0 & \text{ otherwise}\,.
\end{cases}
\end{equation}
To each trajectory up to time $t$, visiting the states $\s_0,
\s_1, \dots, \s_n$  (listed in chronological order, $n$ being a
random integer) and jumping along the (unoriented) edges indexed respectively by $i_0, i_1, \dots, i_{n-1}$,
 we associate the functions $W^{(i)}(t) $, $i: 1\leq
i \leq m$, defined as
\begin{equation}
W^{(i)}(t) := \sum _{k=0}^{n-1} w^{(i)}(\s_k, \s_{k+1} )\d_{i_k=i}\,.
\end{equation}
Shortly, every time the walker jumps along an $i$--edge, the function  $W^{(i)}$ increases of the weight $w^{(i)}(\s, \s')$, $\s$ and $\s'$
being the initial state and the final state of the jump.

\medskip

In what follows, given a distribution $\nu$ on $\cS$  we denote by
$\bbP_\nu$ and $\bbE_\nu$ the probability measure and  the
expectation w.r.t.  the Markov chain $X_t$  with knowledge of the
crossed edges, starting  with initial distribution $\nu$. If $\nu =
\d_\s $, we simply write
 $\bbP_\s$ and $\bbE_\s$ instead of $\bbP_\nu$ and $\bbE_\nu$.

 We point out that for any $\l_1, \dots, \l_m\in \bbR$, the
moment generating  function  $\bbE_\s \left[  e^{-\sum _{i=1}^m
\l_i
 W^{(i)}(t)}\right] $ is finite. Indeed, all the weights $w^{(i)} (\s,
 \s')$ are finite, while by a simple coupling argument one gets that
 the random number of jumps in the time
 interval $[0,t]$ is stochastically dominated by a suitable Poisson
 variable, which has finite  moment generating function.

Similarly to \cite{LS}[Section 2.3] (with the exception that here we
consider arbitrarily initial distributions and we have dropped a
technical  assumption), the following holds:
\begin{theorem}\label{fluttuante} Given a distribution $\nu$ on $\cS$,
for each $\l_1, \dots, \l_m\in \bbR$ the limit
\begin{equation}\label{orologio} e(\l_1, \dots, \l_m):= \lim_{t \to \infty}
-\frac{1}{t} \ln \bbE_\nu \left[ e^{-\sum _{i=1}^m \l_i
W^{(i)}(t)}\right]
\end{equation}
exists,  is finite and does not depend on $\nu$.
 Moreover, the following fluctuation relation is valid:
 \begin{equation}\label{fr}
 e(\l_1, \dots, \l_m)= e(1-\l_1, \dots, 1-\l_m) \,.
 \end{equation}
\end{theorem}
The proof is given in Subsection \ref{integ_pierpi_proof}.  Due to
the theory of large deviations (cf. \cite{V}, Section 1.5 in
\cite{JQQ} and references therein), the  Legendre--Fenchel transform
$$I(z_1, \dots, z_m ):= \sup _{\l_1, \dots, \l_m \in \bbR}\Big\{
 e (\l_1, \dots, \l_m)- \sum_{i=1}^m \l_i z_i \Big\}
$$
is convex, lower--semicontinuous and non--negative. Moreover, it has
compact level sets and satisfies $\inf _{z \in \bbR^m} I(z)=0$. In
addition, the random vector $W^{(i)}(t)/t $ satisfies w.r.t.
$\bbP_\nu$ a large deviation principle with rate function $I$
(independent from $\nu$). As discussed in \cite{LS}, the fluctuation
relation \eqref{fr} can be restated in terms of the rate function
$I$ as
\begin{equation}\label{frbis}
I(z_1, \dots, z_m)- I(-z_1, \dots, -z_m)=-\sum _{i=1}^m z_i \,,
\qquad \forall z \in \bbR^m\,.
\end{equation}
Physical implications of the above relations in the steady state
are discussed in \cite{LS}.

\smallskip

We point out that the above Theorem \ref{fluttuante} leads to an infinite family of
fluctuation relations \eqref{fr}, parameterized by the decompositions $\cL=
\sum _{i=1}^m \cL^{(i)}$. Only some of them refer  to relevant
physical quantities. 
Some physically relevant cases are discussed in \cite{LS}, the
fundamental one corresponds to  $m=1$, in this case $W^{(1)}$ is
related with the  entropy production of the system. As explained in
Section \ref{corrente}, by a suitable choice of the decomposition
$\cL= \sum _{i=1}^m \cL^{(i)}$ one immediately recovers from the
above theorem
 the fluctuation theorems for
currents obtained in \cite{AG2}.

\section{Fluctuation theorem for currents}\label{corrente}

The dynamical evolution of several physical and chemical systems is already  described
 by a random walk on a finite graph, with multiple edges and no loop. Here, we can apply Theorem \ref{fluttuante}
 to this context, obtaining an equivalent fact which is simply  formulated in a different language. Moreover, we re--derive the  fluctuation
 theorem for currents with respect to a fundamental set of oriented cycles obtained in \cite{AG2} and consider a similar problem when working with more
 general sets of oriented cycles (Theorem  \ref{basis} below).


\smallskip

We start now  with a finite connected  unoriented graph $\cG$ whose
vertexes are given by the states in $\cS$.   $\cG$ can have multiple
edges but  no loop, i.e. between two distinct vertexes $\s , \s'$
there can be several edges, while there is no edge from a state to
itself. We denote by  $\cE_o$ ($o\,=\,$oriented) the set of edges of
$\cG$ with orientation (each edge in $\cG$ has two possible
orientations and therefore corresponds to two oriented edges in
$\cE_o$).
Given $\ell\in \cE_o$, we denote by $\bar \ell $ the edge obtained
from $\ell$ by inverting its orientation. Moreover,
we  write $\ell_i$ and $ \ell_f$ for the states in $\cS$ such that
$\ell $ goes from $\ell_i$ to $\ell_f$ (initial and final states).
It is convenient to assign a canonical orientation to each
unoriented edge in $\cG$. We denote by $\cE_c$ ($c\,=\,$canonical)
the set of canonically oriented edges  of $\cG$. Note that each
oriented edge in $\cG$ is given by $\ell$ or $\bar \ell$, for some
$\ell \in \cE_c$.
\medskip

We fix a family of positive numbers $\{ k(\ell)  \}_{\ell \in \cE_o}$
and consider the continuous time  random walk $X_t$ on $\cG$,
jumping  along the edge $\ell$ (following  the associated
orientation) with probability   rate $k(\ell)$. In particular, the
Markov generator of $X_t$ is given by
\begin{equation}
 \cL f( \s) = \sum _{ \ell \in \cE_o: \s = \ell_i  }
k(\ell)\left[ f(\ell_f) - f(\s) \right]\,.
\end{equation}
Since $\cG$ is a connected graph and since the constants $k(\ell)$
are positive, $X_t$ belongs to the class of Markov chains on $\cS$
introduced in Section \ref{pierpi}.


\smallskip

 To each oriented edge $\ell \in \cE_o$, we associate the weight
 \begin{equation}\label{matteo}  w(\ell)= \ln
 k(\ell) / k(\bar\ell) \end{equation} and we define $N_\ell (t)$ as the
number of times   the random walk jumps along  $\ell$ minus
the number of times the random walk jumps along   $\bar
\ell$, up to time $t$. In words, $N_\ell(t)$ is the flux along the oriented edge $\ell$.
Considering the sum decomposition
$\cL= \sum _{\ell \in \cE_c} \cL^{(\ell)}$ where $\cL^{(\ell)}$ is the Markov generator
associated only to the jumps along the oriented edges $\ell$ and $\bar \ell$, it is trivial to
 check that Theorem \ref{fluttuante} implies the following fact:

\begin{fact}\label{basilare}
For any initial distribution $\nu$ and   for any family $\{\l _\ell:
\ell \in \cE _c\}$ of real numbers, the limit
\begin{equation}\label{fr_nonna_q}
q\bigl( \{\l _\ell: \ell \in \cE_c \} \bigr):= \lim_{t \to \infty}
-\frac{1}{t} \ln \bbE_\nu \left[ e^{-\sum _{\ell \in \cE_c}  \l_\ell
 N_\ell (t)} \right]
\end{equation}
exists, is finite and does not depend on $\nu$. Moreover, it holds
\begin{equation}\label{voce_q}
q\bigl( \{\l _\ell: \ell \in \cE_c \} \bigr)= q\bigl( \{w(\ell)-\l _\ell:
\ell \in \cE_c \} \bigr)\,.
\end{equation}
\end{fact}
We point out that, when  $\nu$ is  the stationary distribution, the
above result equals formula (84) in \cite{AG2}.

\subsection{Fluctuation theorem for currents w.r.t. a fundamental set of oriented cycles}\label{prelimare_SCH}
We recall here some basic concepts concerning the oriented cycles in a finite graph \cite{H,SCH,B}.
\cite{SCH}  represents a very concise reference.

\smallskip

An {\sl oriented  cycle } $\cC$ in $\cG$ can be described   by a string of
oriented edges $(b_1,\dots, b_k)$,  $b_i \in \cE_o $, such that  the vertex in which
$b_i$ enters equals the vertex from which  $b_{i+1}$ exits (with the
convention $k+1=1$). We convey  that  the strings $(b_1, \dots , b_k)$ and $(b_i,
b_{i+1}, \dots, b_k, b_1, \dots, b_{i-1})$ identify the same
oriented cycle $\cC$.
 When disregarding the orientation of $\cC$, we
call it simply cycle.

We fix  a {\sl maximal tree} (also called {\sl spanning tree}) $T$ on $\cG$ \cite{SCH} , i.e. $T$ is a
unoriented subgraph of $\cG$, containing all vertexes of $\cG$ and
containing no cycle. The  edges $\ell \in \cE_c$ which do not belong
to $T$ when disregarding orientation  are called {\sl chords} of
the maximal tree $T$. We enumerate the edges in  $\cE_c$ as $\ell_1,
\ell_2, \dots, \ell_n$ where  $\ell_1, \dots, \ell_s$ are the
chords associated to $T$. In particular, $s$ is the number of chords
of $\cG$. It is simple to prove that $s$ equals $e-v+1$, where $e$ is the number of edges of $\cG$ and $v$ the number of vertexes. In
particular, $s$ does not depend on the choice of the maximal tree and will be referred in what follows as {\sl chord number}.

  When adding to $T$ a chord $\ell_i $
(disregarding its orientation), one obtains a graph containing a
unique cycle. We give to this cycle the  orientation induced by the
chord $\ell_i$ and call $\cC_i$ the resulting oriented cycle. The
set $\{\cC_1, \dots, \cC_s\}$ is called a {\sl fundamental set of
oriented cycles}. We now explain the origin of the name. Given an oriented cycle $\cC$ and an oriented
edge $\ell \in \cE_o$, we denote by $S_\ell (\cC)$ the number of times $\ell$  appears in $\cC$
minus the number of times the reversed edge $\bar \ell$ appears in $\cC$. Then
\begin{equation}\label{somma_cicli}
 \cC= \sum _{j=1}^s S_{\ell_j}( \cC) \cC_j\,.
\end{equation}
The meaning of the above identity is clarified by the following definition:
\begin{definition}\label{quadroverde}
Given oriented cycles $\cC_1, \dots, \cC_k,\, \cC_1' , \dots, \cC_r'$ and given real numbers
$a_1, \dots , a_k, \, a_1' ,\dots, a_r'$ we set
\begin{equation}\sum _{i=1}^k a_i \cC_i = \sum _{j=1}^r a_j'\cC_j
\; \Longleftrightarrow \;  \sum _{i=1}^k a_i  S_\ell(\cC_i) = \sum
_{j=1}^r a_j'S_{\ell}(\cC_j') \; \;\; \forall \ell \in \cE_o\,.
\end{equation}
\end{definition}


We point out that, given a fundamental set $\{\cC_1, \dots,\cC_s\}$,
the oriented edges appearing in $\cC_j$   are all distinct, that
$\ell $ and $\bar \ell$ cannot both appear in $\cC_j$ and that, for
each chord $\ell_i$,  it holds $\cS_{\ell_i} (\cC_j)= \d_{i,j}$ for
$1\leq i,j \leq s$.

\smallskip

Recall the definition of the weight \eqref{matteo}. Given a generic oriented cycle $\cC$ we define its affinity as
 $\cA(\cC)= \sum _{\ell} w(\ell)$, where the sum is among the edges $\ell$ visited by the cycle with the proper orientation.
Alternatively, one can set
\begin{equation}\cA( \cC)= \sum _{\ell \in \cE_c} S_\ell (\cC) w(\ell)\,.
\end{equation}
We recall that in the above definition one can  replace $w(\ell)$ by the ratio between  the
 local flux along $\ell$ and the local
flux along $\bar \ell$  w.r.t. the steady state or any other
probability measure  $\nu$ on the vertex set, giving positive
measure to each state, i.e.  $w(\ell)$ can be replaced by  $\ln
\bigl[ \nu(\ell _i) k(\ell)/ \nu (\ell_f) k(\bar \ell) \bigr]$. This
leads to an equivalent definition of cycle affinity.

\smallskip


 Recall the definition of the random variable $N_\ell (t)$
given just after \eqref{matteo}. $N_\ell (t)$ represents the flux
along the oriented edge $\ell$ (its generalized derivative is the
current).
 In \cite{AG2}, the authors have
proved the following result (which trivially leads to a fluctuation
theorem for large deviations), that we state here for a general
initial distribution (only  steady states are considered in
\cite{AG2}):

\begin{fact}\label{robots}
For any initial distribution $\nu$ and   for any  $\l_1, \dots,
\l_s\in \bbR $, the limit
\begin{equation}\label{fr_nonna_AG}
Q( \l_1, \dots, \l_s  ):= \lim_{t \to \infty} -\frac{1}{t} \ln
\bbE_\nu \left[ e^{-\sum _{i=1}^s  \l_i N_{\ell_i} (t)} \right]
\end{equation}
exists, is finite and does not depend on $\nu$. Moreover, it holds
\begin{equation}\label{voce_AG}
Q( \l_1, \dots, \l_s  )= Q \bigl( \cA( \cC_1)- \l_1, \dots, \cA( \cC_s)-
\l_s)\,.
\end{equation}
\end{fact}
 The above Fact \ref{robots}  restricted to the
steady state  corresponds to formula  (39)  in \cite{AG2}. As
observed there, under Schnakenberg's conditions (see the
Introduction), this formula leads immediately to the fluctuation
theorem for currents restated in terms of macroscopic forces (see
page 124 in \cite{AG2}). In Section \ref{chiudo} we show that Fact
\ref{robots} is a simple corollary of Fact \ref{basilare}. This was
already pointed out  in \cite{AG2}, by different arguments.

\subsection{Fluctuation theorem for generalized currents associated to a basis  of oriented cycles}\label{FT_basis}
In this section, 
  we extend the fluctuation
theorem to  currents referred to (what we call) a  basis of oriented
cycles in $\cG$. This answers a natural conceptual question. In
addition, in applications, one can have identified a nice basis of
oriented cycles and desire to work with it, without looking for a
nice fundamental set of oriented cycles (see Section
\ref{motorino}).
 First we fix some general concepts,
applying some ideas of linear algebra to oriented cycles.

We say that  the family of oriented cycles $\cC_1, \dots , \cC_k$ in $\cG$ generates all the oriented cycles if for each
oriented cycle $\cC$ there exist real numbers $a_1, \dots, a_k$ satisfying
$\cC= \sum _{i=1}^k a_i \cC_i$ (recall Definition \ref{quadroverde}).
 Such a {\sl generating set}  is called {\sl basis} if
it is minimal, in the sense that it does not contain a smaller generating subfamily. It is simple to check that minimality can be replaced by
 independence, i.e.
there are not  constants $a_1, \dots, a_k$ (not all zero) such that
$\sum _{i=1}^k a_i \cC_i=\emptyset $, where $\emptyset$ denotes the
degenerate cycle with $S_\ell (\emptyset)=0$ for each oriented edge
$\ell$. In the last case, we also say that $\cC_1, \dots, \cC_k$ are
{\sl independent}. Lemma \ref{semplice} in the integrating Section
\ref{integ_FT_basis} will report some very intuitive facts
concerning bases, generating sets and independent sets. Here we
recall that all bases  have cardinality given by the chord number
$s$, that any fundamental set of oriented cycles is a basis (the
opposite implication is false, see Section \ref{integ_FT_basis}) and
that for each oriented cycle $\cC$ the coefficients $a_1, \dots,
a_s$ in the decomposition
\begin{equation}\label{decomposion}
\cC= \sum_{i=1}^s a_i \cC_i\,,
\end{equation}
where $\cC_1, \dots, \cC_s$ is a basis,
  are univocally determined. If the basis is not a fundamental set of oriented cycles the coefficients $a_1, \dots, a_s$ do not have a simple
geometric characterization as in \eqref{somma_cicli}. As showed  in Section \ref{integ_FT_basis},
 if $\cC_1, \dots, \cC_s$ is a basis and $\cC_1', \dots, \cC_s'$
is a fundamental set of oriented cycles with associated chords $\ell_1, \dots, \ell_s$,  then it holds
\begin{equation}\label{frontale_bisq}
 (a_1, \dots, a_s)= \bigl( S_{\ell_1}(\cC), \dots, S_{\ell_s}(\cC) \bigr) B^{-1}\,, \qquad B=(B_{ij})_{1\leq i,j\leq s} \,,\; B_{ij}=
S_{\ell_j}( \cC_i)
 \,.
\end{equation}

\smallskip

Let us come back to our random walk on the graph $\cG$. We want to associate
 to each trajectory up to time $t$ an oriented cycle  $\cC_t$. We
do it as follows.
Of the  trajectory   up to time $t$
 we record the sequence of oriented bonds $(b_1, \dots, b_n)$, $b_i \in \cE_o$,
which the walker  moves along, one after the other ($n$ is the number of jumps performed up to time $t$).
 We denote this path as {\sl reduced trajectory}. Recall the definition of the function  $N_\ell(t)  $
 given  after \eqref{matteo}.
 Note that $N_\ell(t)$  depends only on the reduced trajectory, since
 $N_{\ell }(t)= \sum _{i=1}^n [ \bbI (  b_i=\ell )- \bbI(b_i=\bar
 \ell)]$.
Given two different states $\s, \s'$ we fix a  path $\g_{\s,\s'}$ in $\cG$ from $\s$ to $\s'$ specifying only the visited states and the
edges along which the path evolves (there is no knowledge of jump times), i.e. $\g_{\s, \s'} $ is represented by a sequence of oriented
bonds. We then introduce a random oriented cycle $\cC_t $ on $\cG$ as follows:
\begin{definition}\label{def_Ct}
Fixed a family $\{ \g_{\s, \s' }\}_{\s\not = \s' }$, let  $(b_1,
\dots, b_n)$ be the reduced trajectory from $X_0$ to $X_t$.  If
$X_0=X_t$  we define $\cC_t$ as the oriented cycle given by the
reduced trajectory itself,
  otherwise
we define $\cC_t$ as the cycle $(b_1, \dots, b_n, c_1, \dots, c_r)$
where $\g_{X_0, X_t}= (c_1, \dots, c_r)$.
\end{definition}
We point out that
\begin{equation}\label{ciclico} N_\ell (t) = S_\ell (\cC_t)+ O(1) \qquad \forall \ell \in \cE_o
\end{equation}
where the errors $O(1)$ are uniformly bounded  as $\ell$ varies  in
$\cE_o$ and $t$ varies in  $[0,\infty)$.

\medskip

From now on, we refer to a basis $\cC_1, \dots, \cC_s$ and a family  $\{ \g_{\s, \s' }\}_{\s\not = \s' }$, fixed once and for all.
Due to the above discussion, we know that  the random numbers $a_1(t), \dots, a_s (t)$
such that
\begin{equation}\label{strada}
 \cC_t = \sum _{i=1}^s a_i (t)\cC _i,
 \end{equation}
exist and are univocally determined. Note that these random numbers
depend not only from the basis but also on the paths $\g_{\s, \s'}$.
On the other hand, choosing other paths $\g_{\s, \s'}$ would change
the random numbers $a_i(t)$ of quantities uniformly bounded as  $t$
varies in $[0, \infty)$ (as a
 byproduct of \eqref{frontale_bisq} and \eqref{ciclico}). We point out that, by \eqref{somma_cicli} and \eqref{ciclico}, if our basis is a fundamental set
 of oriented cycles with associated chords $\ell_1, \dots , \ell_s$,
  then it holds $a_i(t)= N_{\ell_i}(t)+O(1)$. In particular, a part an error of order
 $O(1)$, the number $a_i(t)$ is simply the flux along the chord $\ell_i$, which equals the time--integrated current along $\ell_i$.
 Due to this special case, we call   $a_i(t)$  {\sl generalized time--integrated current}.

 \smallskip

 We can finally state our fluctuation theorem:
\begin{theorem}\label{basis}  Fix a basis $\cC_1, \dots, \cC_s$ and paths $\bigl\{\g_{\s, \s'}\bigr\}_{\s\not=\s'}$. Then,
for any initial distribution $\nu$ and   for any  $\l_1, \dots,
\l_s\in \bbR $, the limit
\begin{equation}\label{uffa1}
Q_b( \l_1, \dots, \l_s  ):= \lim_{t \to \infty} -\frac{1}{t} \ln
\bbE_\nu \left[ e^{-\sum _{i=1}^s  \l_i a_i (t)} \right]
\end{equation}
exists, is finite and does not depend on $\nu$. Moreover, it holds
\begin{equation}\label{uffa2}
Q_b( \l_1, \dots, \l_s  )= Q_b ( \cA(\cC_1)- \l_1, \dots, \cA(\cC_s)-
\l_s)\,.
\end{equation}
\end{theorem}
The index $b$ refers to the term ``basis''.
Due  to the fact that $a_i(t)= N_{\ell_i}(t)+O(1)$ when the basis is a fundamental set of oriented cycles,
  the above result extends Fact \ref{robots}.
The proof  of Theorem \ref{basis} is given in Section \ref{brodo}.


\section{Applications to molecular motors moving along a polymer}\label{motorino}

 In this section we consider a molecular motor \cite{Ho} moving along a polymer (mathematically, a one dimensional
 periodic environment), under the effect of an external force $f$, using  chemical energy to produce mechanical
 work. We write $\D \mu$ for the chemical potential difference associated to ATP hydrolysis. We denote by $x_t$ the position of the molecular motor along the filament and by $z_t$ the number of
 hydrolyzed  ATPs minus the number of synthesized ATPs such that $z_0=0$.
 We  want to apply Theorem \ref{basis} to prove in full generality the
 relation \eqref{fatto} described in the Introduction.

   \smallskip

The above fluctuation relation \eqref{fatto} has been proved in
\cite{LLM0,LLM1} for a special model with two chemical states. The
method followed
 there
relies on the manipulation of  $2\times 2$ matrices (since only two
chemical states are considered there). On the other hand,  keeping
in mind the above observations concerning macroscopic forces and
affinities, one would expect the above symmetry \eqref{fatto} to be
universal. This is indeed what we prove below as confirmation of the
large flexibility of the fluctuation theorem referred to  oriented
cycles. To deal this problem in full generality one has to work with
a very large graph.

\smallskip
First of all we need to fix a  discrete kinetic model. The natural
modeling of molecular motors is by continuous   models called {\sl
ratchet models} \cite{JAP,Re}. Methods to derive a discrete model
from the continuous ratchet model have been developed in full
generality (cf. \cite{XWO} and reference therein), thus leading to
the following class of kinetic models. The model is a
continuous--time Markov chain  with state space $\bbZ\times \G$,
$\G$ being the  finite set of the chemical states of the motor, with
master  equation
\begin{multline}
\partial _t p_t(x,\s)  =
\o^\s _{x-1,x} p_t (x-1, \s) + \o ^\s_{x+1,x}  p_t (x+1, \s)
 - \bigl[ \o^\s _{x,x-1}+\o ^\s _{x,x+1}\bigr ] p_t(x,\s)  \\
 + \sum _{\s'\not=\s}\bigl[ \o^x  _{\s',\s}p_t(x,\s')  -    \o^x_{\s,
\s '} p_t(x,\s) \bigr] \,.
\end{multline}
Above, $p_t(x, \s)$ denotes the probability of the motor to be  at
site $x \d$ ($\d$ being the length of spatial unity in the
discretization)
 at time $t$ in chemical state $\s$.   $\o_{x,x'}^\s$ is the probability rate for a jump from the position $x\d $ to
  $x'\d $ (with $|x-x'|=1$)
 if the chemical  state is  $\s$ , and $\o_{\s,\s'}^x$ is the probability rate for a chemical transition from $\s$ to $\s'$
 if the mechanical coordinate  is $x$.
 When considering motor proteins
moving along polymeric filaments, the above rates are periodic
function in $x$, with the same period.

In order to keep information about the ATP consumption, it is
convenient to enrich the above discrete model, extending the state space and
distinguishing between active and passive chemical transitions (i.e.
related to hydrolysis/synthesis of ATP or  to thermal noise) \cite{PJAP}. The
new state of the system (motor plus environment) is now described by
the triple $(x, \s, z ) \in \bbZ\times \G\times \bbZ$, where $z$
denotes the algebraic number of hydrolyzed ATP molecules (in the
sense that $z$ increases of one unity for each ATP hydrolysis and
decreases of one unity for each ATP synthesis). We write $\o^{x,
l}_{\s,\s'}$ for the probability rate of a chemical transition from
$\s$ to $\s'$ at the  mechanical state $x$ with the consumption of
$l=-1, 0, 1$ ATP molecules. Hence, $\o^{x, l}_{\s,\s'}$ is  the
probability rate for the jump $(x, \s, z)\to (x, \s', z+l)$. The
other possible jumps are the mechanical ones $ (x, \s, z)\to (x',
\s, z)$, $|x-x'|=1$, having probability rate $\o^{\s}_{x,x'}$. The
resulting master equation is the following:
\begin{multline}
\partial _t p_t(x,\s,z)  =
\o^\s _{x-1,x} p_t (x-1, \s,z) + \o ^\s_{x+1,x}  p_t (x+1, \s,z)
 - \bigl[ \o^\s _{x,x-1}+\o ^\s _{x,x+1}\bigr ] p_t(x,\s,z)\\
+ \sum_{l=-1,0,1} \sum _{\s'\not=\s}\bigl[ \o^{x,l} _{\s',\s}p_t(x,\s',z-l)  -    \o^{x,l}_{\s,
\s '} p_t(x,\s,z) \bigr] \,.
\label{master}
\end{multline}
In what follows, we suppose that all the rates are positive. In principle, this is not a restriction since positive but very
small rates correspond to very unlikely transitions.

\smallskip
 Recall that $f$ denotes the load force and $\D \mu$ the chemical potential difference associated to ATP hydrolysis.
We write $V_\s$ for the potential energy of the molecular motor in
the chemical state $\s$ due to the interaction with the polymer.
Again, the functions $V_\s$, $\s\in \G$, are periodic functions of
the same period.
 The energy associated to the state $(x, \s, z)$  is given by
\begin{equation} E(x,\s, z)= V_\s (x\d) - fx\d - z \D \mu + \text{ const. }
\end{equation}
 Then,   the detailed balance condition reads
\begin{equation}\label{dbc}
\begin{cases}
\o _{x,x+1}^\s\,/\, \o _{x+1,x}^\s  = \exp \left\{ - \b   \bigl[ V_\s (x\d+\d)- V_\s (x\d) \bigr] +\b \d f \right\}\,,\\
 \o_{\s, \s'} ^{x, l}\,/\,  \o _{\s', \s}^{x, -l}  = \exp \left \{ - \b\bigl[ V_{\s'} (x\d)- V_\s (x\d) \bigr]+ \b l \D \mu
 \right\}\,.
\end{cases}
\end{equation}

We can now formulate our result:
\begin{fact}\label{applicazione}
Recall that  $x_t$, $z_t$ denote  the position at time $t$ of the
molecular motor and the total  ATP consumption up to  time $t$.
Denote by $\overline{f\b}$  and $\bar x_t$  the adimensional
quantities $f \b \times [1m]$ and $x_t /[1m]$ ($[1m]$ being the
length unit).
 Then for any given initial configuration
$(x_0, z_0=0, \s_0)$ and for any constants
 $\l, \g \in \bbR$
the following limit exists:
\begin{equation}\label{roma1}
 \vartheta ( \l, \g) := \lim _{t \to \infty} - \frac{1}{t} \ln \bbE _{(x_0, z_0=0, \s_0)} \Big[ e^{- \l \bar x_t- \g z_t } \Big]\,.
\end{equation}
Moreover, it holds
\begin{equation}\label{meta}
 \vartheta ( \l, \g) = \vartheta (\overline{f\b} - \l, \b\D\mu -\g)\,, \qquad \l , \g \in
 \bbR\,.
 \end{equation}
\end{fact}
 We point out that the model presented in
\cite{KW,LLM0,LLM1} is indeed a coarse--graining approximation of
the model we have described above when considering two chemical
states and suitable potentials.
\subsection{Proof of Fact \ref{applicazione}}
 In order to apply Theorem  \ref{basis} we need to
work with a finite state space. This is obtained by a suitable
projection.

Recall that the filament is periodic. We define $N$ as the number of
$\d$--units contained in a spatial period. This implies in
particular
\begin{equation}\label{periodico}
\begin{cases}
V_\s (x\d)= V_\s (x\d+N\d)\,,
&\o ^{x,l}_{\s, \s'}= \o ^{x+N,l}_{\s,\s'} \,, \\
\o _{x+1, x}^\s= \o _{x+N+1, x+N}^\s \,,& \o _{x,x+1}^\s = \o_{x+N,
x+N+1}^\s \,.
\end{cases}
\end{equation}
We write $\bbZ_N$ for the quotient space $\bbZ/ N \bbZ$ given by integers modulo $N$. Of course, sums in $\bbZ_N$ are  modulo $N$.
Due to \eqref{periodico}
both the potentials and the jump rates can be thought of with the spacial parameter $x$ varying in $\bbZ_N$.


We introduce a finite graph $\cG$ with multiple edges and no loop, indicating the canonical orientation of edges.
 To this aim it is convenient
to label the states in $\G$ as $\s_1, \s_2, \dots, \s_m$ ($m = |\G|$). Then,
 the  vertexes  of $\cG$ are  given by the pairs $(x, \s_i)$, $x \in \bbZ_N$ and
$1\leq i \leq m $. We put an oriented edge $e_{x,x+1}^{\s_i} $
 from state $(x,\s_i)$ to  state $(x+1,\s_i)$.  Moreover,  given $1\leq i <j \leq m$, we put three oriented  edges
from state  $(x,\s_i)$ to state  $(x, \s_j)$ labeled by the index
$l=-1, 0,1$ and we call them $e^{x,l}_{\s_i, \s_j}$. See figure
\ref{lati}. Note that $\cG$ has $N  m $ vertexes  and $Nm+ 3N m
(m-1)/2$ edges. In particular, its chord number  is
\begin{equation}\label{numerocorde}
s= 3N m (m-1)/ 2 +1\,.
\end{equation}

\begin{figure}[!ht]
    \begin{center}
       \psfrag{a}{$(x,\s_i)$}
      \psfrag{b}{$(x+1,\s_i)$}
       \psfrag{c}{$e^{\s_i}_{x,x+1}$}
    \psfrag{d}{$e^{x,-1}_{\s_i, \s_j}$}
    \psfrag{e}{$e^{x,0}_{\s_i, \s_j}$}
   \psfrag{f}{$ e^{x,1}_{\s_i, \s_j}$}
     \psfrag{g}{$(x, \s_j)$}
 \includegraphics[width=7cm]{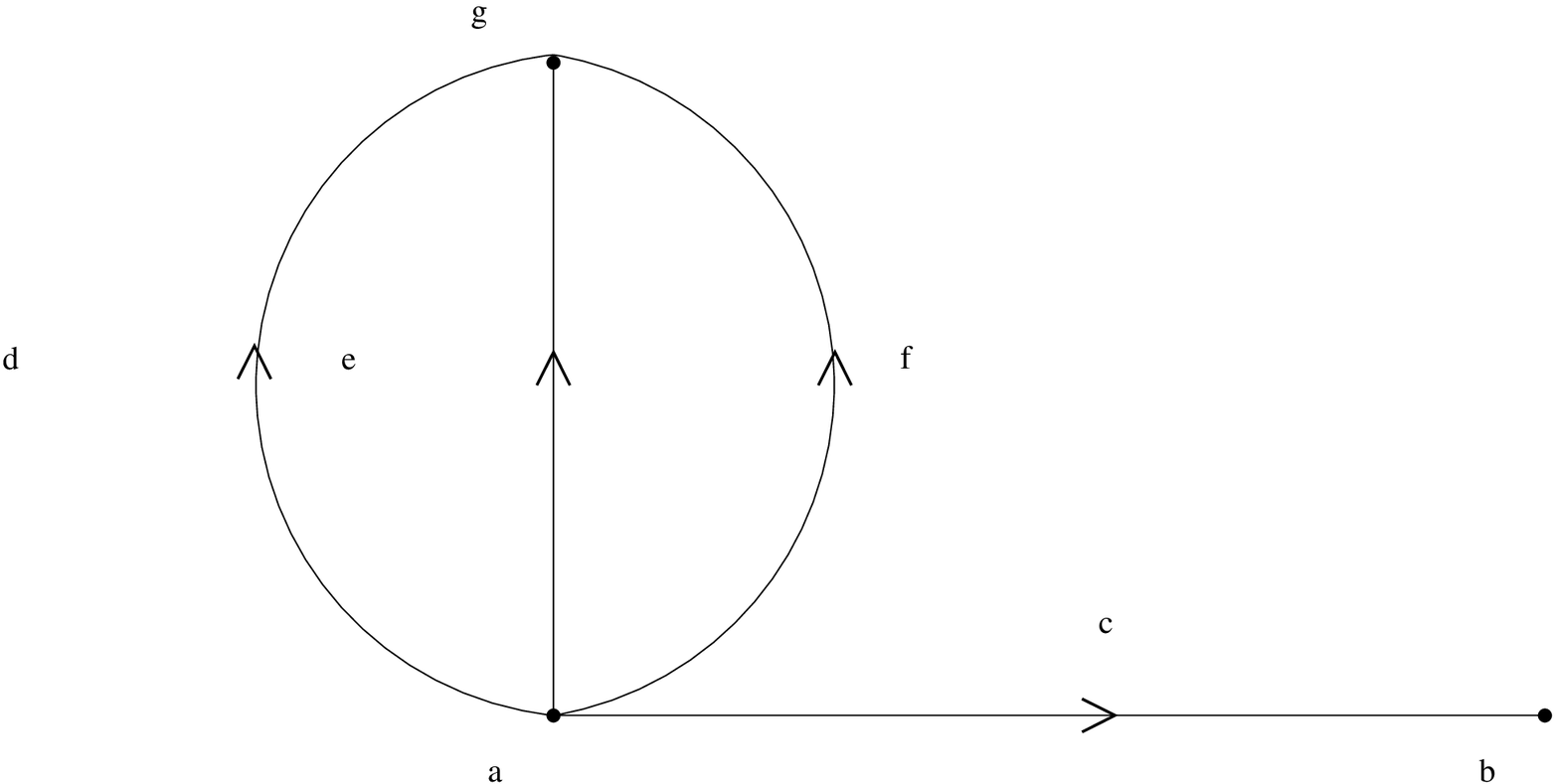}
     \caption{The edges $e_{x,x+1}^{\s_i}$, $\;e^{x,-1}_{\s_i, \s_j}$, $\;e^{x,0}_{\s_i, \s_j}$, $\;e^{x,1}_{\s_i, \s_j}$.}
    \label{lati}
    \end{center}
  \end{figure}

\smallskip

 We now consider the continuous time   random walk  $X_t$ on the graph $\cG$, where the probability rate
for a jump along the  edge $e_{x,x+1}^{\s_i} $ equals
$\o_{x,x+1}^{\s_i} $, while the probability rate of the reversed
jump is $\o_{x+1,x}^{\s_i} $. Similarly, given $1\leq i < j \leq m$
and $l=-1,0,1$, the probability rate for a jump along the edge
$e^{x,l}_{\s_i, \s_j}$ equals $\o^{x,l}_{\s_i, \s_j}$, while the
probability rate of the reversed jump is $\o^{x,-l}_{\s_j, \s_i}$.
Recall that, given an oriented  edge $\ell$, we write $N_\ell(t)$
for the number of times the random walk jumps along the edge $\ell$
minus the number of times the  random walk jumps along the reversed
edge $\bar \ell$, up to time $t$. Trivially, the above random walk
$X_t$ can be coupled with the Markov chain considered in the
previous subsection getting the identities
\begin{equation*}
 x_t/\d = \sum _{x\in \bbZ_N} \sum _{1\leq i \leq m}  N_{e_{x,x+1}^{\s_i}}(t)\,,\qquad \qquad
 z_t = \sum _{x\in \bbZ_N} \sum _{1\leq i <j \leq m}\Big[ N_{  e ^{x, 1}_{\s_i, \s_j} }(t)  - N_{e ^{x, -1}_{\s_i, \s_j} }(t)\Big]\,.
\end{equation*}
This fact, together with \eqref{fr_nonna_q}, implies that the limit \eqref{roma1} exists and is finite, hence the function
$\vartheta (\l, \g)$ is well defined.

\smallskip

It remains to prove \eqref{meta}.
 To this aim recall how we have associated the cycle $\cC_t$ to the  trajectory of the random walk  up to time $t$
 (see Definition \ref{def_Ct}). Trivially, a part an error
of order one uniformly in $t \in [0,\infty)$, the above identities
can be rewritten as
\begin{equation}\label{codroipo}
\begin{cases}
  x_t /\d= \sum _{x\in \bbZ_N} \sum _{1\leq i \leq m}  S_{e_{x,x+1}^{\s_i}}(\cC_t)+O(1) \,,\,\\
  z_t = \sum _{x\in \bbZ_N} \sum _{1\leq i <j \leq m}\bigl[ S_{  e ^{x, 1}_{\s_i, \s_j} }(\cC_t)  - S_{e ^{x, -1}_{\s_i, \s_j} }(\cC_t)\bigr]+O(1)\,.
  \end{cases}
\end{equation}
We give now two different applications of Theorem \ref{basis}. In
the first case  we exhibit a nice basis which is not a fundamental
set, in the second case we exhibit a spanning tree leading to a nice
fundamental set.

\subsubsection{A nice basis} We  introduce a basis of oriented cycles on $\cG$.
We call $\cC_0$ the cycle described by the ordered family of
oriented edges
 $$\cC_0=\bigl( e_{1,2}^{\s_1}, e_{2,3}^{\s_1}, \dots, e_{N-1, N}^{\s_1}, e_{N,1}^{\s_1}\bigr)\,.$$
Given $x\in \bbZ_N$ and $1\leq i <j \leq m$, we introduce the three  oriented cycles (see figure \ref{cicli_facce})
$$\cC_{\s_i, \s_j} ^{x, -1}= \bigl( e_{\s_i, \s_j}^{x,0} , \bar e_{\s_i, \s_j}^{x,-1}\bigr)
\,,\qquad \cC_{\s_i, \s_j} ^{x, 1}= \bigl( e_{\s_i, \s_j}^{x,1} , \bar e_{\s_i, \s_j}^{x,0}\bigr)\,,
$$
$$ \cC^{x,x+1}_{\s_i, \s_j}= \bigl( e _{x,x+1}^{\s_i}, e_{\s_i, \s_j}^{x+1,0}, \bar e _{x,x+1}^{\s_j}, \bar e_{\s_i, \s_j}^{x,0} \bigr)\,.
$$

\begin{figure}[!ht]
    \begin{center}
        \psfrag{a}{$ \cC_{\s_i, \s_j} ^{x, -1} $}
      \psfrag{b}{$ \cC_{\s_i, \s_j} ^{x, 1} $}
         \psfrag{c}{$\cC^{x,x+1}_{\s_i, \s_j}$}
         \psfrag{d}{$(x,\s_i)$}
     \psfrag{e}{$(x+1,\s_i)$}
      \psfrag{f}{$(x+1, \s_j) $}
       \psfrag{g}{$(x, \s_j)$}
 \includegraphics[width=10cm]{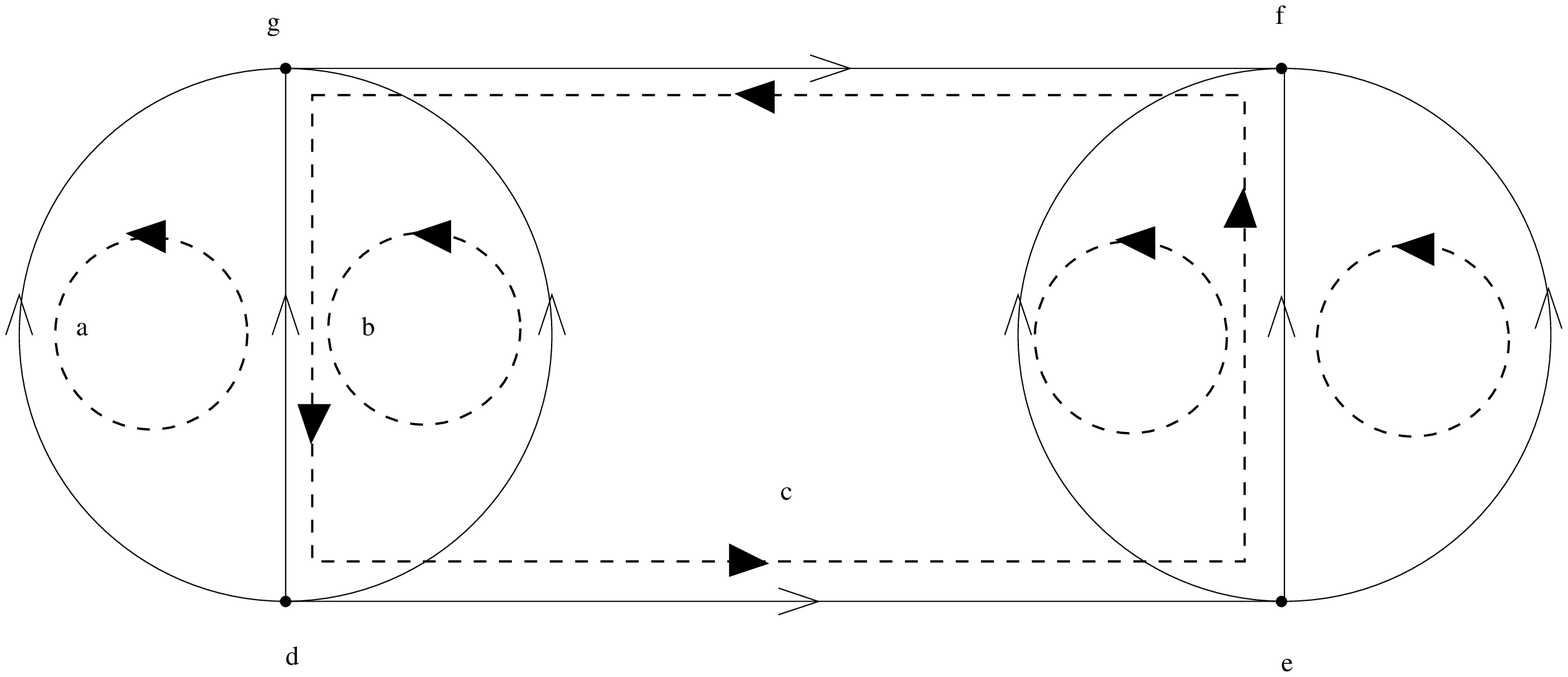}
      \caption{The cycles    $\; \cC_{\s_i, \s_j} ^{x, -1}$, $\; \cC_{\s_i, \s_j} ^{x, 1}$, $\; \cC^{x,x+1}_{\s_i, \s_j}$.}
    \label{cicli_facce}
    \end{center}
  \end{figure}
Consider now, for each $x$, the complete graph on the set $\{(x,
\s_i)\,:\, 1\leq i \leq m \}$ with canonically oriented edges
$e^{x,0}_{\s_i,\s_j}$, $1\leq i< j\leq m$ (see figure \ref{acqua}).
Trivially, the edges $e^{x,0}_{\s_1,\s_2}, e^{x,0}_{\s_2,\s_3},
\dots, e^{x,0}_{\s_{m-1},\s_m}$ form a spanning tree. We call
$\cC^{x,0}_{\s_i,\s_j}$ the oriented cycle associated to the chord
$e^{x,0}_{\s_i, \s_j}$, with $(i,j)\not \in\{(1,2), (2,3),\dots,
(m-1,m)\}$. The orientation of $\cC^{x,0}_{\s_i,\s_j}$ agrees with
the one of the chord $e^{x,0}_{\s_i, \s_j}$.

\begin{figure}[!ht]
    \begin{center}
        \psfrag{a}{$ (x,\s_1) $}
      \psfrag{b}{$  (x,\s_2)$}
         \psfrag{c}{$(x,\s_3)$}
         \psfrag{d}{$(x,\s_4)$}
 \includegraphics[width=3cm]{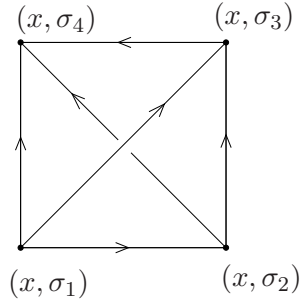}
      \caption{The complete graph with edges $e^{x,0}_{\s_i,\s_j}$, $1\leq i< j\leq
m:=4$.}
    \label{acqua}
    \end{center}
  \end{figure}
\begin{lemma}\label{raidue}
Consider the family of oriented cycles \begin{itemize}
\item[(i)]
$\cC_0$,
\item[(ii)]  $\cC_{\s_i, \s_j} ^{x, - 1}$ with $x\in \bbZ_N$ and $1\leq
i<j\leq m$,
\item[(iii)] $\cC_{\s_i, \s_j} ^{x,  1}$ with $x\in \bbZ_N$ and $1\leq
i<j\leq m$,
\item[(iv)]
$\cC_{\s_i, \s_j} ^{x,  0}$ with $x\in \bbZ_N$, $1\leq i< j \leq m$
and $(i,j)\not\in \{(1,2),(2,3),\dots,(m-1,m) \}$,

\item[(v)] $\cC^{x,x+1}_{\s_i, \s_{i+1}}$ with $x\in \bbZ_N$ and
$1\leq i \leq m-1$.
\end{itemize}
Then  the above family of oriented cycles is a basis.
\end{lemma}
\begin{proof}
The oriented cycles of type (i),(ii),...,(v) are respectively $1$,
$Nm(m-1)/2$, $Nm(m-1)/2$, $Nm(m-1)/2-N(m-1)$, $N(m-1)$. It follows
then that the above family of oriented cycles has cardinality
$3Nm(m-1)/2+1$, which is the chord number $s$ given  in
\eqref{numerocorde}.

 By point (v) in Lemma
\ref{semplice}, in order to show  that it is a basis it is enough to
prove independence. Suppose to have a linear combination of the
above oriented cycles which equals the degenerate cycle $\emptyset$.
First we observe that the edge $\bar e^{x,-1}_{\s_i,\s_{j}}$ belongs
only to $\cC^{x,-1}_{\s_i,\s_{j} }$, the reversed edge does not
belong to any oriented cycle in the family. By applying the operator
$S_\ell$ with $\ell=\bar e^{x,-1}_{\s_i,\s_{j}}$, and then varying
$x,i,j$, we get that the cycles of type (ii) do not appear in the
linear combination. The same holds for cycles of type (iii) (think
to the edges $e^{x,1}_{\s_i,\s_{j}}$). Keeping in mind this
consideration, the edge $e^{x,0}_{\s_i,\s_{j}}$ with  $x\in \bbZ_N$,
$1\leq i<j \leq  m$ and $(i,j)\not\in \{(1,2),(2,3),\dots,(m-1,m)
\}$, appears only in the oriented cycle $\cC^{x,0}_{\s_i,\s_{j} }$
of type (iv) while its reversed edge does not appear in any cycle of
the family (recall the properties of a fundamental set of oriented
cycles). Note that edges of the forms $e^{x,0}_{\s_i,\s_{j}}$ or
$\bar e^{x,0}_{\s_i,\s_{j}}$ appear in the oriented cycles of type
(v) but it must be $j=i+1$.

At this point, we know that the linear combination involves only
$\cC_0$ and the oriented cycles of type (v), i.e. we have
\begin{equation}\label{militare}
\a_0 \cC_0+ \sum _{x \in \bbZ_N} \sum _{1\leq i <  m}
\a^{x,x+1}_{\s_i, \s_{i+1}} \cC^{x,x+1}_{\s_i,
\s_{i+1}}=\emptyset\,.
\end{equation}
Applying to both members  the operator $S_\ell$ with $\ell
=e_{x,x+1}^{\s_i}$ ($x \in \bbZ_N$ and $1\leq i \leq m$), we get for
any $x\in \bbZ_N$
\begin{equation} \begin{cases} \a_0+
\a^{x,x+1}_{\s_1,\s_{2}}=0 & \text{  }\\
\a^{x,x+1}_{\s_i,\s_{i+1}}=\a^{x,x+1}_{\s_{i+1},\s_{i+2}} & \text{
if } 1\leq i \leq m-2 \,,\\
\a^{x,x+1}_{\s_{m-1},\s_{m}}=0\,.
\end{cases}
\end{equation}
The above system trivially implies that all coefficients in
\eqref{militare} are zero. Hence, the oriented cycles in the above
family are independent.
\end{proof}
We point out that the above basis is not a fundamental set of
oriented cycles. Indeed, each edge of $\cC_0$ belongs also to some
oriented cycle of type (v), hence $\cC_0$ can not contain any chord.

Let us now compute the affinities associated to the  oriented cycles
in our basis.  Due to the detailed balance relations \eqref{dbc}, we
have the following weights:
\begin{equation}
\begin{cases}
w\bigl( e^\s_{x,x+1}\bigr)= -  \b \bigl[ V_\s (x\d+\d)- V_\s (x\d) \bigr] + \b\d f \,,\\
w\bigl( e^{x,l}_{\s,\s'} \bigr)=
  - \b\bigl[ V_{\s'} (x\d)- V_\s (x\d) \bigr]+\b l \D \mu\,.
\end{cases}
\end{equation}
 Hence,
\begin{align}
&  \cA \bigl( \cC_0\bigr)= \b Nf \d\,, \qquad \cA\bigl(\cC_{\s_i,
\s_j} ^{x, -1}\bigr)=  \b \D \mu\,, \qquad \cA \bigl(\cC_{\s_i,
\s_j} ^{x, 1}
\bigr)=  \b \D \mu\,, \\
 & \qquad \qquad \qquad \cA\bigl(\cC_{\s_i, \s_j}
^{x, 0}\bigr)= 0\,, \qquad  \cA \bigl( \cC^{x,x+1}_{\s_i,
\s_j}\bigr)=0 \nonumber
\end{align}
(note that conservative force fields never appear in the cycle
affinities). We write \begin{multline}
\cC_t= a_0(t) \cC_0+ \sum_{x
\in \bbZ_N} \sum _{1\leq i <j \leq m} \Big (
a_{\s_i, \s_j} ^{x, -1} (t) \cC_{\s_i, \s_j} ^{x, -1}
+ a_{\s_i, \s_j} ^{x, 1} (t)   \cC_{\s_i, \s_j} ^{x, 1}\Big)+\\
\sum_{x\in \bbZ_N} \sum_{\substack{1\leq i <j \leq m\\ (i,j) \not
\in \{ (1,2), \dots,(m-1,m)\} } }a_{\s_i, \s_j} ^{x, 0} (t)
\cC_{\s_i, \s_j} ^{x, 0}
+ \sum_{x \in \bbZ_N}\sum _{1\leq i <m } a^{x,x+1}_{\s_i, \s_{i+1}}
(t) \cC^{x,x+1}_{\s_i, \s_{i+1}}\,.
\end{multline}
By applying \eqref{codroipo} to the above decomposition and setting
$\bar \d= \d/[1m]$, we get that
\begin{equation}
 \l \bar x_t   + \g z_t =\l N \bar \d a_0(t) + \g \sum _{x \in \bbZ_N} \sum _{1\leq i < j \leq m} \Big[
a_{\s_i, \s_j} ^{x, -1} (t) +a_{\s_i, \s_j} ^{x, 1} (t)
\Big]+O(1)\,. \end{equation}
The above identity together with Theorem \ref{basis} implies that
\begin{multline}
 \vartheta (\l, \g)=
 Q_b( \cC_0 \to \l N \bar \d \,,\; \cC_{\s_i, \s_j} ^{x, -1} \to \g\,,\;
\cC_{\s_i, \s_j} ^{x, 1} \to \g\,, \; \cC_{\s_i, \s_j}
^{x, 0}\to 0\,,\; C^{x,x+1}_{\s_i, \s_j} \to 0)=\\
Q_b( \cC_0 \to \b  N f\d-\l N\bar \d \,,\; \cC_{\s_i, \s_j} ^{x, -1}
\to \b \D \mu -\g\,,\; \cC_{\s_i, \s_j} ^{x, 1} \to \b \D \mu-\g\,,
\;
\cC_{\s_i, \s_j} ^{x, 0}\to 0\,,\; C^{x,x+1}_{\s_i, \s_j} \to 0)=\\
\vartheta ( \overline{\b f}-\g, \b \D \mu-\g)\,.
\end{multline}
This concludes the proof of Fact \ref{meta}.

\subsubsection{A nice fundamental set} The above arguments can be
applied  also to  fundamental sets of oriented cycles. As example of
spanning tree, one can take the tree given by the edges
$e^{x,0}_{\s_1,\s_i}$ with $2\leq i \leq m$ and $x\in \bbZ_N$, and
$e^{\s_1}_{x,x+1}$ with $1\leq x\leq N-1$ (see figure
\ref{figurina1}). Alternatively, one can take as spanning tree the
one given by the edges $e^{1,0}_{\s_i,\s_j}$ with $(i,j)\in \{(1,2),
(2,3), \dots, (m-1,m)\}$, and $e^{\s_i}_{x,x+1}$ with $1\leq i \leq
m$ and $1\leq x \leq N-1$. In both cases, the associated fundamental
set is tractable.  The advantage here is that one does not have to
check that the associated fundamental set is a basis, since this is
automatically. We leave the details to the interested reader.

\begin{figure}[!ht]
    \begin{center}
        \psfrag{a}{$ (1,\s_1) $}
      \psfrag{b}{$  (1,\s_2)$}
         \psfrag{c}{$(1,\s_3)$}
         \psfrag{d}{$(1,\s_4)$}
  \psfrag{e}{$x=1$}
\psfrag{f}{$x=2$} \psfrag{g}{$x=3$} \psfrag{h}{$x=1$}
 \includegraphics[width=12cm]{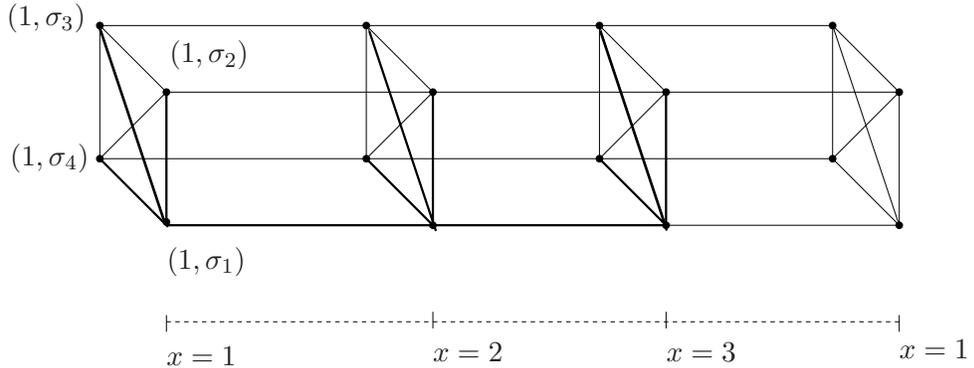}
      \caption{Example of spanning tree where $N=3$, $m=4$. The spanning tree is in boldface,
      edges $e^{x,\pm 1  }_{\s_i,\s_j}$ have been omitted.}
    \label{figurina1}
    \end{center}
  \end{figure}

\section{Mathematical integration to Section  \ref{pierpi}}\label{integ_pierpi}

 The  decomposition of the generator $\cL$ in terms of the rates $k^{(i)}(\s,\s')$ described in Section \ref{pierpi}
 is naturally
associated to the following construction of the Markov chain $X_t$.
First we observe that  the evolution of the Markov chain is
univocally determined by the sequence   $\s_0 , \s_1, \s_2,\dots$ of
the visited states (listed in chronological order) and by the times
$T_n $, $n \in \bbN$, of  the $n$--th jump  (with the convention
$T_0=0$). In order to specify these last  quantities, we consider a
family of independent exponential random variables
$\cT^{(i,n)}_{\s,\s'}$, parameterized by the integers
  $n \in \bbN$ and $i:1 \leq i \leq m$ and by the  ordered pairs
  $(\s, \s')$  of distinct elements in $\cS$. The random variables
$\cT^{(i,n)}_{\s,\s'}$ are all defined on a same probability space $(\O, \cF, \bbP)$,
  $\cT^{(i,n)}_{\s, \s'}$ is an exponential variable with  parameter (i.e. $1/\text{mean}$)
  $k^{(i)}(\s,\s')$. Note that $\cT^{(i,n)}_{\s,\s'}=\infty$ if
  $k^{(i)}_{\s,\s'}=0$.

\smallskip

 The  Markov chain $X_t$ starting in $\s_0$
  can be inductively constructed  as follows as function on the space $\O$: known the $n$--th jump time $T_n$ and
   the $n$--th visited state
  $\s_n$, we set $\cT^{(n)}= \min _{i, \s'} \cT^{(i,n)}_{\s_n, \s'}$.
  Then $T_{n+1}:= T_n + \cT^{(n)}$, i.e. after the $n$--th jump the Markov chain remains at
  $\s_n$ for a time $\cT^{(n)}$, after that it jumps to the state
  $\s'$ such that for some  $i$ it holds  $\cT^{(n)}=   \cT^{(i,n)}_{\s_n, \s'}$.
  As well known (cf.
 \cite{N}[Theorem 2.3.3]), the index $i$ and the state $\s'$ such that
 $\cT^{(n)}=   \cT^{(i,n)}_{\s_n, \s'}$ are univocally determined
 a.s.



\smallskip

Recall the interpretation of the Markov chain as random walk on the  graph $\cG$, given in Section \ref{pierpi}.
Known $\s_n$ and $T_n$, we can think
that the walker moves at time $T_{n+1}=T_n+ \cT^{(n)}$ to the state
$\s'$ along the $i$--edge if $\cT^{(n)}=   \cT^{(i,n)}_{\s_n, \s'}$. In addition,
recall that in Section \ref{pierpi}, given the initial distribution  $\nu$ on $\cS$, we have defined $\bbP_\nu$ as the probability
measure describing the evolution of the Markov chain keeping knowledge of the crossed edges.
We point out that  $\bbP_\nu$ is a probability measure on the  space  $\O$,  where all the exponential variables
$\cT^{(i,n)}_{\s,\s'}$ are defined.


\subsection{Proof of Theorem \ref{fluttuante}}\label{integ_pierpi_proof}
 The proof follows the
strategy stated in \cite{LS}[Section 2.3], i.e.  one has simply to
extend the proof given in \cite{LS}[Section 2.1] for $m=1$.
The only difference is that one has to keep in mind the special
construction of the Markov chain $X_t$ given above.
 We give the proof  for completeness and also to  show
that the assumptions required in \cite{LS}[Section 2.3] are
redundant.

We fix $\l_1, \dots, \l_m$ and set $g_t (\s)=\bbE_\s \Big[
e^{-\sum _{i=1}^m \l_i
 W^{(i)}(t)}\Big]$. It is well known that the function $[0,\infty)
 \ni t \to g_t (\s) \in (0,\infty)$ is differentiable.
We think  of
 $X_t$ as a random walk  on the graph $\cG$  with multiple edges and no loop as
 explained in Section \ref{pierpi}. Fixed $\e>0$ we introduce the
 following events: $A$ is the event that in the time interval $[0,
 \e]$ the walker does not jump, $B^{(i)}_{\s, \s'} $ is the event
 that in the time interval $[0,
 \e]$ the walker makes a unique jump and this jump takes place  along the $i$--labelled edge
 from $\s$ to $\s'$, while $C$ denotes the event that neither  $A$ nor any
  event $B^{(i)}_{\s, \s'}$   takes place.  By the
Markov property we can write \begin{multline}\label{batteria}
 g_{t+\e}(\s)  = \bbE_\s
\left[ e ^{-\sum _{i=1}^m \l_i W^{(i)}(\e) }g_t (\s_\e) \right]
=\bbE_\s \left[ e ^{-\sum _{i=1}^m \l_i W^{(i)}(\e) }g_t (\s_\e) ;A
\right]  \\
+ \sum _{\s':\s'\not = \s}\sum _{i=1}^m\bbE_\s \left[ e ^{-\sum
_{j=1}^m \l_j W^{(j)}(\e) }g_t (\s_\e); B^{(i)}_{\s, \s' }\right]
+\bbE_\s \left[ e ^{-\sum _{i=1}^m \l_i W^{(i)}(\e) }g_t (\s_\e) ;C
\right] \\ = g_t (\s) \bbP_\s (A)+  \sum _{\s':\s'\not = \s}\sum
_{i=1}^m e^{-\l_i w^{(i)}_{\s, \s'} } \bbP_\s ( B^{(i)}_{\s,
\s'})g_t(\s')+ \bbE_\s \left[ e ^{-\sum _{i=1}^m \l_i W^{(i)}(\e)
}g_t (\s_\e) ;C \right] \,.
\end{multline} Using that $\bbP_\s (A)= 1-r(\s)\e + o(\e)$, $\bbP_\s (B^{(i)}_{\s,
\s'} )=k^{(i)}(\s, \s')\e + o(\e)$ and $\bbP_\s (C)=o(\e)$, where
$$r(\s):=\sum _{\s': \s' \not = \s} k(\s, \s')= \sum  _{\s': \s'
 \not = \s} \sum _{i=1}^m k^{(i)}(\s, \s') \,,$$
 the above expansion  \eqref{batteria} implies that
$$ \frac{d g_t}{dt} (\s)= \sum _{\s': \s' \not = \s} \Big[
\sum _{\substack{i: 1\leq i \leq m \\
k^{(i)}(\s,\s')
>0}} k^{(i)} ( \s,\s')^{1-\l_i} k^{(i)} (\s', \s)^{\l_i}\Big] g_t (\s')-
r(\s) g_t (\s)=: \bigl(\cL_{\underline \l } g_t \bigr) (\s)\,.
$$
 The above differential
equation and the fact that $g_0(\s)=1$ for any $\s$ imply that $g_t
(\s) = \bigl( e^{t \cL _{\underline \l} } g_0\bigr) (\s)=\sum
_{\s'\in \cS} \bigl[e^{t \cL_{\underline \l}}\bigr]_{\s, \s'}$. In
particular, we can write
\begin{equation}\label{ferrodastiro}
\bbE_\nu \left[ e^{-\sum _{i=1}^m \l_i W^{(i)}(t)}\right]= \bigl(
\nu, e^{t \cL _{\underline \l} } \mathbf{1} \bigr)\,,
\end{equation}
where $(\cdot, \cdot )$ denotes the Euclidean scalar product in
$\bbR ^{\cS}$ and $\mathbf{1}$ denotes the vector with all entries
equal to $1$.

We point out that the  $\cL_{\underline \l}$ has off-diagonal
nonnegative  entries.  In particular, for $a$ large enough, $M:=\cL
_{\underline \l} + a \bbI$ has  nonnegative entries and one can
apply to it the Perron--Frobenious theorem (cf. \cite{JQQ}[Theorem
1.5.4]).
 From the fact
that $X_t$ is an irreducible Markov chain, we derive that
$\cL_{\underline \l}$ is irreducible (i.e. given $\s\not= \s'$ there
exists a path $\s=\s_0, \s_1, \dots, \s_n =\s'$ such that
$[\cL_{\underline \l}]_{\s_j,\s_{j+1}}>0$ for all $j:0\leq j <n$)
and as a  consequence the same conclusion holds for $M$. Applying
now the Perron--Frobenius theorem we conclude that $M$ has a maximal
eigenvalue which is simple and is associated to an eigenvector $v$
having all positive entries. Hence,  the same conclusion holds for
$\cL_{\underline \l}$. We call $-e(\underline \l)$ the maximal
eigenvalue of $\cL_{\underline \l}$ associated to $v$. We write
$v_j$ for the $j$--th entry  of $v$ and set  $v_{\rm{max}}= \max _j
v_j$, $v_{\min}= \min _j v_j$. Since both the vectors $v$, $\nu$ and
the matrix  $e^{t \cL _{\underline \l} }= e^{-ta } e^{t M}$ has only
nonnegative entries,  we conclude that \begin{multline}  v_{\max}
^{-1} e^{-t e (\underline \l)} (\nu ,v) = v_{\max} ^{-1} \bigl( \nu
, e^{t\cL_{\underline \l} } v\bigr) \leq \bigl(  \nu , e^{t
\cL_{\underline \l} } \mathbf{1}\bigr)= \bbE _\nu \left[ e ^{-\sum
_{i=1}^m \l_i W^{(i)}(t) } \right]  \\ \leq v_{\min} ^{-1} \bigl(
\nu  , e^{t\cL_{\underline \l} } v\bigr)=v_{\min} ^{-1}e^{-t
e(\underline \l )} (\nu,v)\,.
\end{multline}
Since  $(\nu, v)>0$ the above estimate trivially implies
\eqref{orologio}.

 Since $\cL^*_{(1-\l_1, \dots , 1-\l _m)}= \cL_{(\l_1, \dots, \l_m
 )}$ (the l.h.s. denotes the transposed matrix, recall the definition of $\cL_{\underline \l}$) and since $\cL^*_{(1-\l_1, \dots , 1-\l
 _m)}$ and $\cL_{(1-\l_1, \dots , 1-\l _m)}$ have the same
 eigenvalues, we immediately get \eqref{fr}. \qed

\section{Derivation of Fact \ref{robots} from Theorem
\ref{fluttuante} }\label{chiudo}

Recall that we have labeled the canonical oriented edges in $\cE_c$
as $\ell_1, \ell_2, \dots, \ell_s,$ $ \ell_{s+1}, \dots, \ell_n$,
where $\ell_1, \dots, \ell_s$ are the chords. Due to Definition
\ref{def_Ct}, the representation \eqref{ciclico}  and due to
\eqref{fr_nonna_q} we get
\begin{multline}\label{fr_nonna_bis}
q\bigl( \{\l _\ell: \ell \in \cE_c \} \bigr):= \lim_{t \to \infty}
-\frac{1}{t} \ln \bbE_\nu \left[ e^{-\sum _{\ell \in \cE_c}  \l_\ell
  N_\ell (t)} \right]=\\
  \lim_{t \to \infty}
-\frac{1}{t} \ln \bbE_\nu \left[ e^{-\sum _{\ell \in \cE_c}  \l_\ell
  S_\ell (\cC_t)} \right]\,.\end{multline}

 Due to \eqref{somma_cicli} (recall the definition \eqref{matteo} of the weight $w(\ell)$)  we can write
\begin{multline*} \sum _{i=1}^s \l_i S_{\ell_i}(\cC_t)= \sum _{i=1}^s \l_i S_{\ell_i}(\cC_t)+
\sum _{i=s+1}^n w(\ell_i) S_{\ell_i}(\cC_t)-\sum _{i=s+1}^n
w(\ell_i)  S_{\ell_i}(\cC_t)\\ = \sum _{i=1}^s \l_i
S_{\ell_i}(\cC_t)+
\sum _{i=s+1}^n w(\ell_i)  S_{\ell_i}(\cC_t)- \sum _{i=s+1}^n \sum _{j=1}^s  w(\ell_i)  S_{\ell_j} (\cC_t) S_{\ell_i} (\cC_j)\\
=\sum _{i=1}^s \bigl[ \l_i  - \sum _{a=s+1}^n  w (\ell_a) S_{\ell_a}(\cC_i)\bigr] S_{\ell_i}(\cC_t)+
\sum _{i=s+1}^n w(\ell_i)  S_{\ell_i}(\cC_t)\,.
\end{multline*}
The above identity and \eqref{fr_nonna_bis} imply that the limit \eqref{fr_nonna_AG} exists, is finite and equals
$$ Q(\l_1, \dots, \l_s)= q\Big( \{  \l_i  - \sum _{a=s+1}^n  w (\ell_a) S_{\ell_a}(\cC_i)   \}_{1\leq i \leq s},
 \{w(\ell_i) \} _{s+1\leq i \leq n} \Big)\,.
$$
Due to the fluctuation relation \eqref{voce_q}, the expression in the r.h.s. equals
\begin{equation}\label{cassano}
 q\Big( \{  -\l_i  +w(\ell_i) + \sum _{a=s+1}^n  w (\ell_a) S_{\ell_a}(\cC_i)   \}_{1\leq i \leq s},
 \{0\} _{s+1\leq i \leq n} \Big)\,.
 \end{equation}
As already observed, given $i: 1\leq i \leq s$, $S_{\ell_j}
(\cC_i)=\d_{i,j}$ for $1\leq i,j \leq s$. In particular, $w(\ell_i)
+ \sum _{a=s+1}^n  w (\ell_a) S_{\ell_a}(\cC_i)$ simply equals the
affinity $\cA\bigl(\cC_i\bigr)$. Combining this observation with the
above identities,
 we get that
$$ \eqref{cassano}= q\Big( \{ \cA(\cC_i) -\l_i    \}_{1\leq i \leq s},
 \{0\} _{s+1\leq i \leq n} \Big)= Q( \{ \cA(\cC_i)- \l_i \}_{1\leq i \leq s})\,.
 $$
This completes the proof of Fact \ref{robots}. \qed

\section{Mathematical integration to Subsection \ref{FT_basis} }\label{integ_FT_basis}

We list some simple  properties concerning linear combinations of
oriented cycles:
\begin{lemma}\label{semplice}
 The following holds:
 \begin{itemize}
 \item[(i)] any generating set of oriented cycles has cardinality at least the $\text{chord number }s$,
 \item[(ii)] any independent set of oriented cycles has cardinality at most $s$,
\item[(iii)] any basis has cardinality $s$,
\item[(iv)]   any  generating set of oriented cycles   of cardinality $s$ is a basis,
 \item[(v)]
any independent set of oriented cycles of cardinality $s$ is a basis,
\item[(vi)]
given a basis $\cC_1, \dots, \cC_s$ and given an oriented  cycle $\cC$,  the coefficients $a_1, \dots , a_s$
such that $\cC= \sum_{i=1}^s a_i \cC_i$  are univocally determined,
\item[(vii)]
 any fundamental  set of oriented cycles is a basis.
\end{itemize}
\end{lemma}
\begin{proof}
We first prove (i), taking a generating set of oriented cycles
$\cC_1, \dots , \cC_k$. Then we
 fix a fundamental set of oriented cycles $ \cC_1', \dots,  \cC_s'$ and call $\ell_1, \dots, \ell_s$ the associated chords. We introduce the $k \times s$
 matrix $B$ with $B_{i,j}= S_{\ell _j} (\cC_i)$. Due to \eqref{somma_cicli} it holds $ \cC_i = \sum _{j=1}^s B_{ij} \cC_j'$ for
  all $i: 1\leq i \leq k$. On the other hand, by definition of generating set, there exists a $s \times k$ matrix $A$ such that $\cC_j'= \sum _{i=1}^k A_{ji}
  \cC_i$ for all $j: 1\leq j \leq s$. Then, for all  $r: 1\leq r \leq s$, it holds
  $$ \d_{j,r} = S_{\ell_r} ( \cC_j')= \sum _{i=1}^k A_{ji} S_{\ell_r} (C_i)= \sum _{i=1}^k A_{ji} B_{ir}\,.$$
The above relations can be rewritten as $\bbI = A B$, where $\bbI$ denotes the $s \times s$ identity matrix. This trivially implies that $B$ must have rank
at least $k$, hence $k\geq s$.

\smallskip

In order to prove (ii) suppose that $\cC_1, \dots, \cC_k$ is an independent set of oriented cycles. Taking a fundamental set
$ \cC_1', \dots,  \cC_s'$  as above,  for a suitable $k \times s $ matrix we can write  $\cC_i = \sum _{j=1}^s B_{ij} \cC_j'$ for
  all $i: 1\leq i \leq k$. If $k >s$, then we could find a non trivial zero linear combination of the rows  of $B$, i.e. not all zero
  coefficients $a_1, \dots, a_k$ such that $ \sum _{i=1}^k a_i B_{ij}=0$ for all $j: 1\leq j \leq s$. Since $
  \sum _{i=1}^k  a_i S_{\ell}( \cC_i) = \sum _{i=1}^k\sum _{j=1}^s a_i B_{i,j} S_\ell (\cC_j')$, this fact
trivially implies that
  $ \sum _{i=1}^k a_i \cC_i=\emptyset$, in contradiction with the  hypothesis of independence. This completes
  the proof of point (ii).
We can directly prove point (v). To this aim suppose that $k=s$.
Then, as argued above, independence implies that the above matrix
$B$ has rank $s$, hence it is invertible. Calling $A:=B^{-1}$,
since   $\cC_i = \sum _{j=1}^s B_{ij} \cC_j'$ , this trivially
implies that $ \sum _{i=1}^s A_{ki } \cC_i= \cC_k'$, for all $k :
1\leq k \leq s$. In particular, the family $\cC_1, \dots, \cC_k$
generates a fundamental set of oriented cycles which generates all
oriented cycles  by \eqref{somma_cicli}.

\smallskip
Point (iii) is an immediate consequence of (i) and (ii). To prove
point (iv), let $\cC_1, \dots \cC_s$ be a generating set of
cardinality $s$. If these cycles were dependent, then we could find
a generating subset of oriented cycles, in contradiction with point
(i). Property (vi) follows from the definition of basis, while
property (vii) follows from (iv) and \eqref{somma_cicli}.
\end{proof}

We point out that not any basis is a fundamental set of oriented cycles.
 Consider figure 5. Vertexes are numbered from $1$ to $6$.
 Consider the oriented cycles $C_1=1 \to 2 \to 5 \to 6 \to 1$, $\cC_2 = 2\to 3 \to 4 \to 5 \to 2 $, $\cC_3 = 3\to 4 \to 6 \to 1\to 3$ and
 $\cC_4 = 4 \to 5 \to 6 \to 4$ (above we have indicated the visited vertexes, in order of visit). This  family is a basis. Indeed, it has cardinality given by the chord number and moreover it trivially generates
 the fundamental set of oriented cycles
 associated to the maximal tree given by the edges $\{1,2\}$,
 $\{1,3\}$,
 $\{1,6\}$, $\{4,6\}$, $\{5,6\}$
 depicted in the figure by boldface (with exception of the edge $\{1,3\}$).
On the other hand $\cC_1, \dots, \cC_4$ do not form a fundamental set of oriented cycles since each  edge of $\cC_4$ belongs (with the same or with
opposite orientation) to some other oriented cycle $\cC_i$, $i\not=4$. Hence, no edge of $\cC_4$ could be the chord associated to a hypothetical maximal tree.

\begin{figure}[!ht]
    \begin{center}
   \psfrag{a}{$1$}
      \psfrag{b}{$2$}
       \psfrag{c}{$3$}
     \psfrag{d}{$4$}
  \psfrag{e}{$5$}
  \psfrag{f}{$6$}
 \includegraphics[width=4cm]{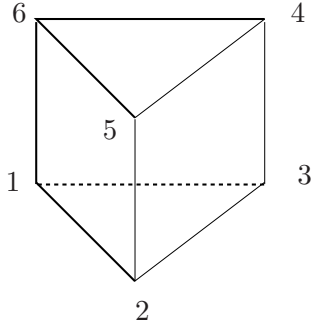}
\caption{Graph for an example of basis which is not a fundamental
set.}
    \end{center}
  \end{figure}

We conclude this section with the proof of Theorem \ref{basis}:

\subsection{Proof of Theorem \ref{basis} }\label{brodo}
We fix a fundamental set of oriented cycles $ \cC_1', \dots,  \cC_s'$ and call $\ell_1, \dots, \ell_s$
the associated chords.
Recall that the orientations of $\cC_i'$ and $\ell_i$ agree.  We introduce the $s \times s$ matrix $B$ defined as
$B_{ij}= S_{\ell_j} ( C_i)$. Then by \eqref{somma_cicli} it holds $\cC_j= \sum _{i=1}^s B_{ji} \cC_i'$.
 We point out that the matrix $B$ is invertible (see the arguments used in the proof of Lemma \ref{semplice}).

Given an oriented  cycle $\cC$ we can express the coefficients $a_i$ in the representation $\cC= \sum _{i=1}^s a_i \cC_i$ by
means of the fundamental set of oriented cycles and associated chords as follows. Applying the operator $S_{\ell_j}$ to the identity
$\cC= \sum _{i=1}^s a_i \cC_i$,
 we get
\begin{equation}
 S_{\ell_j} ( \cC)= \sum _{i=1}^s a_i S_{\ell_j} ( \cC_i)= \sum _{i=1}^s a_i B_{ij}\,, \qquad j =1, \dots, s\,.
 \end{equation}
This implies that
\begin{equation}\label{frontale}
 (a_1, \dots, a_s)= \bigl( S_{\ell_1}(\cC), \dots, S_{\ell_s}(\cC) \bigr) B^{-1}\,.
\end{equation}

Writing $< \cdot, \cdot>$ for the scalar product in $\bbR^s$ and
denoting by $S(t)$ the column vector with entries
$S_{\ell_1}(\cC_t), \dots, S_{\ell_s}(\cC_t)$, the above equation
\eqref{frontale} applied to  the oriented cycle $\cC_t$ gives $ \sum
_{i=1}^s a_i(t)\l_i= <S(t), B^{-1} \l>$. In particular, the
expectation  in \eqref{uffa1} can be written as
$$
\bbE_\nu \Big[ e^{- \sum _{i=1}^s \l_i a_i(s)} \Big] = \bbE _\nu \Big[e^{-<S(t), B^{-1} \l>}\Big]\,.
$$
By the above identity, the representation \eqref{ciclico} and Fact
\ref{robots}, we get that the limit in \eqref{uffa1} exists and it
holds $ Q_b (\l)= Q(B^{-1} \l)$. Applying  the fluctuation relation
\eqref{voce_AG} and afterwards again the identity $Q_b ( \cdot) =
Q(B^{-1} \cdot)$,  we get
$$ Q_b (\l)= Q( B^{-1} \l)= Q(\cA - B^{-1}\l)= Q_b ( B( \cA - B^{-1}\l))= Q_b ( B \cA - \l)\,,$$
where $\cA$ denotes the column vector with entries the affinities $\cA (\cC'_1), \dots, \cA ( \cC_s')$.
To conclude, it is enough to note that
\begin{multline*}
 (B \cA)_i = \sum _{j=1}^s B_{ij} \cA( \cC_j')= \sum _{j=1}^s  S_{\ell_j}(\cC_i) \cA ( \cC_j')=
\sum _{j=1}^s S_{\ell_j}(\cC_i)\Big[ \sum _{\ell \in \cE_c} S_\ell (\cC_j') w(\ell)\Big]=\\
\sum _{\ell \in \cE_c} w(\ell) \Big[ \sum _{j=1}^s S_{\ell_j}(\cC_i) S_\ell (\cC_j')]= \sum _{\ell \in \cE_c} w(\ell)
S_\ell ( C_i) = \cA ( \cC_i)\,.
\end{multline*}\qedhere

\bigskip

\bigskip

\noindent {\bf Acknowledgements}: The authors thank L. Bertini and
D. Gabrielli for useful discussions. A.F. acknowledges  the
financial support of the European Research Council through the
``Advanced Grant'' PTRELSS 228032.

\end{document}